\documentclass[11pt]{amsart}

\usepackage[margin=1.05in]{geometry}
\usepackage{amsmath,amssymb,mathtools}
\usepackage{enumitem}
\usepackage{microtype}
\usepackage{xcolor}
\usepackage{hyperref}
\usepackage{comment}

\definecolor{linkblue}{RGB}{25,58,92}

\hypersetup{
  colorlinks=true,
  linkcolor=linkblue,
  citecolor=linkblue,
  urlcolor=linkblue,
  pdftitle={Classification of Maximally Charged Black Holes},
  pdfauthor={Sven Hirsch and Yiyue Zhang}
}

\numberwithin{equation}{section}

\newtheorem{theorem}{Theorem}[section]
\newtheorem{proposition}[theorem]{Proposition}
\newtheorem{lemma}[theorem]{Lemma}

\theoremstyle{definition}
\newtheorem{definition}[theorem]{Definition}

\theoremstyle{remark}
\newtheorem{remark}[theorem]{Remark}

\newcommand{\R}{\mathbb{R}}

\title{Classification of Maximally Charged Black Holes}

\author{Sven Hirsch}
\address{
Department of Mathematics,
Columbia University,
2990 Broadway,
New York, NY 10027, USA
}
\email{sven.hirsch@columbia.edu}

\author{Yiyue Zhang}
\address{
Beijing Institute of Mathematical Sciences and Applications,
Beijing 101408, China
}
\email{zhangyiyue@bimsa.cn}

\begin{document}

\maketitle

\begin{abstract}
We characterize all maximally charged black hole spacetimes in $3+1$ dimensions.
More precisely, given any initial data set $(M,g,k)$ saturating the mass-charge inequality that satisfies the charged dominant energy condition, we show that $(M,g,k)$ must arise from an isometric embedding into a Majumdar-Papapetrou spacetime.
\end{abstract}

\section{Introduction}

Electromagnetic fields play a fundamental role in the physics of black
holes. Accreting black holes are capable of
supporting strong magnetic fields, and magnetic flux near the
event horizon is believed to be a central ingredient in the production of
relativistic jets. The classical Blandford--Znajek mechanism
\cite{BlandfordZnajek}, for instance, extracts rotational energy from a black
hole through electromagnetic fields threading the horizon. More recently,
observations by the Event Horizon Telescope have directly resolved
magnetic-field structures  around M87*
\cite{EHTM87} and Sagittarius A* \cite{EHTSgrA}.

\medskip

Nevertheless, the mathematical properties of charged black holes remain
remarkably little understood, as the vast majority of works assume the
absence of electromagnetic fields. A central conjecture, going back to
Gibbons--Hull \cite{GH}, Tod \cite{Tod}, and
Gibbons--Hawking--Horowitz--Perry \cite{GHHP}
in the early 1980s, proposes a characterization of all maximally charged
black holes. Progress toward this conjecture has a long history, with partial
results due to Bartnik--Chruściel \cite{BartnikChrusciel},
Chruściel--Reall--Tod \cite{CRT},
Khuri--Weinstein \cite{KhuriWeinstein}, Costa \cite{Costa}, Bray-Hirsch-Kazaras-Khuri-Zhang \cite{BHKKZ}, 
and Raulot \cite{Raulot, RaulotRigidity}, among others.

\medskip

We begin by recalling the mass inequality in the charged setting, cf.
\cite{GH,BartnikChrusciel}:

\begin{theorem}
\label{PMT}
Let $(M^3,g,k,E,B)$ be an admissible initial data set\footnote{See Definition \ref{admissible}.}
satisfying the charged dominant energy condition $\mu\ge|J|$, where
 \[
        \mu
        :=
        \frac{1}{2}
        \left(
            R_g-|k|_g^2+(\operatorname{tr}_g k)^2
        \right)
        -|E|_g^2-|B|_g^2,
\qquad
 J
        :=
        \operatorname{div}_g
        \bigl(k-(\operatorname{tr}_g k)g\bigr)
        +2E\times_g B.
    \]
    Then, at every
asymptotically flat end, the ADM energy $\mathcal E$, ADM momentum $P$,
total electric charge $Q_E$, and total magnetic charge $Q_B$ satisfy
\[
    \mathcal E
    \ge \sqrt{|P|^2+Q_E^2+Q_B^2}.
\]
\end{theorem}

The main result of this paper is the corresponding rigidity theorem.

\begin{theorem}
\label{T:main}
Under the assumptions of Theorem \ref{PMT}, suppose that
\[
    \mathcal E^2
    =|P|^2+Q_E^2+Q_B^2.
\]
Then $(M^3,g,k)$ admits an isometric embedding
as a spacelike hypersurface in a Majumdar--Papapetrou\footnote{See Definition \ref{Def:MP}.} spacetime.
\end{theorem}

We now describe the central idea of the proof. Equality in Theorem
\ref{PMT} yields a nontrivial super-covariantly constant spinor $\psi$,
satisfying
\begin{equation}
\label{eq:supercov}
    \nabla_i\psi
    =
    \left(
        -\frac12 k_{ij}e_je_0
        +\frac12 E e_i e_0
        -\frac12 B e_i\chi
    \right)\psi.
\end{equation}
We define the function $N=|\psi|^2$, the vector field $X_i=\langle e_ie_0\psi,\psi\rangle$, and note that $N\ge|X|$.
If $\psi$ is null (in the sense that $N=|X|$ holds everywhere) or timelike (in the sense that $N>|X|$ everywhere), the results follow from previous papers \cite{Tod, CRT, HirschZhang} and it remains to understand the mixed causal type, also compare \cite{HirschZhangAdS, HirschHuang} where similar difficulties were encountered.

\medskip

The key observation is that a super-covariantly constant spinor naturally
produces three distinguished one-forms. Let
\begin{align*}
    Y_i=& \operatorname{Im}\langle e_i \chi\psi,\psi\rangle,\qquad
    \widetilde Y_i= \operatorname{Im}\langle e_i \mathcal J\chi\psi,\psi\rangle,\qquad
    \widehat Y_i= \operatorname{Re}\langle e_i \mathcal J\chi\psi,\psi\rangle,
    \end{align*}
where $\mathcal J$ denotes the quaternionic action on spinors. A direct
consequence of \eqref{eq:supercov} is that the associated one-forms are closed.

\medskip

A separate topological argument shows that $M$ is simply connected.
Consequently these three one-forms admit global potentials and define a map
\[
    T=(T^1,T^2,T^3)\colon M\longrightarrow\mathbb R^3
\]
satisfying
\[
    dT^1=Y^\flat,\qquad
    dT^2=\widetilde Y^\flat,\qquad
    dT^3=\widehat Y^\flat.
\]
The map $T$ has first appeared in Tod's timelike rigidity proof \cite{Tod} and we refer to it as \emph{Tod map} associated with $\psi$.

\medskip

The decisive point is that the spinorial algebra translates into strong
global restrictions on this map. The Jacobian of $T$ is nonnegative, while the asymptotic behavior implies
that $T$ has degree one. 
On the other hand, if the causal type of $\psi$
changes, the degeneracy of the spinorial bilinears forces the Tod map to
develop nontrivial branching to ensure that the local mapping
degree is at least two. 
Consequently $\psi$ cannot change causal type.

\medskip

\noindent
\textbf{AI usage:} We made substantial use of ChatGPT 5.6 Pro and Theorem \ref{T:main} was obtained without significant assistance from the authors. In particular, the central argument regarding the degree of the Tod map was suggested by AI. 

\medskip

\noindent
\textbf{Acknowledgements:} We thank Piotr Chrusciel and Thomas K\"orber for many helpful discussions. YZ was partially supported by NSFC Grant No.\ 12501070 and the startup fund from BIMSA.

\section{Preliminaries}

\begin{definition}\label{Def:MP}
A Majumdar--Papapetrou (MP) spacetime \cite{Majumdar, Papapetrou} $(\overline M,\overline g)$ is defined by $\overline M=\mathbb{R}\times (\mathbb R^3\setminus\bigcup_{i=1}^N \{p_i\})$ and 
$$
\overline g =-U^{-2}dt^2+U^2\delta
$$
where $p_i$ are points in $\mathbb R^3$, $\delta$ is the Euclidean metric and
$U:M^3\to (0,\infty)$ is a positive harmonic function given by
\[
U=1+\sum_{i=1}^{N}\frac{m_i}{|x-p_i|},
\]
for some $m_i> 0$.
\end{definition}

Note that $(\overline M,\overline g)$ is a solution to the Einstein--Maxwell equations with $F=d(U^{-1})\wedge dt$.

\begin{definition}[Weighted function spaces]
\label{Def:weighted}
Let $B\subset\mathbb R^n$ be a ball containing the origin. For $s\in\mathbb N_0$,
$\alpha\in(0,1)$, $p\in[1,\infty)$, and $q\in\mathbb R$, define
\begin{align*}
\|f\|_{C^{s,\alpha}_{-q}(\mathbb R^n\setminus B)}
:={}&
\sum_{|I|\leq s}
\sup_{x\in\mathbb R^n\setminus B}
|x|^{q+|I|}\,|\nabla_I f(x)|
\\
&+
\sum_{|I|=s}
\sup_{\substack{x,y\in\mathbb R^n\setminus B\\
0<|x-y|\leq |x|/2}}
|x|^{q+s+\alpha}
\frac{|\nabla_I f(x)-\nabla_I f(y)|}{|x-y|^\alpha}.
\end{align*}
The corresponding weighted Sobolev norm is
\[
\|f\|_{W^{s,p}_{-q}(\mathbb R^n\setminus B)}
:=
\left(
\sum_{|I|\leq s}
\int_{\mathbb R^n\setminus B}
|x|^{p(q+|I|)-n}
|\nabla_I f|^p\,dx
\right)^{1/p}.
\]
For tensor and spinor fields, all derivatives and pointwise norms are
taken with respect to the Euclidean connection and metric.
\end{definition}
\begin{definition} \label{admissible}
Let $(M^3,g,k,E,B)$ be a smooth, connected, oriented, and complete
initial data set with electric field $E$ and magnetic field $B$. We call
$(M^3,g,k,E,B)$ \emph{admissible} if the following conditions hold:
\begin{enumerate}
    \item There exists a compact subset $K\subset M$ such that
    \begin{equation*}
        M\setminus K
        =\bigsqcup_{a=1}^{n_1} \mathcal{U}_a
        \sqcup
        \bigsqcup_{\ell=1}^{n_2} \mathcal{C}_\ell, \qquad n_1\ge1,
    \end{equation*}     
    where each $\mathcal{U}_a$ is an asymptotically flat end and each
    $\mathcal{C}_\ell$ is a cylindrical end.

    \item For each asymptotically flat end $\mathcal{U}_a$, there exists a
    diffeomorphism
    \[
        \varphi_a:\mathcal{U}_a
        \longrightarrow
        \mathbb{R}^3\setminus\overline{B}
    \]
    such that for some fixed constant $q\in(\frac{1}{2},1)$ and $\alpha\in(0,1)$,
    \[
        g-\delta\in C^{2,\alpha}_{-q},
        \qquad
        k,\;E,\;B\in C^{1,\alpha}_{-q-1},
    \]
where we omit the map $\varphi_a$ in the expressions.
    \item For each cylindrical end $\mathcal{C}_\ell$, there exist a
    diffeomorphism
    \[
        \phi_\ell:\mathcal{C}_\ell
        \longrightarrow (0,\infty)\times S^2,
    \]
   and constants $m_\ell>0$,
    $\mathfrak{e}_\ell,\mathfrak{b}_\ell\in\mathbb{R}$ such that, with
    \[
        g_{\ell}:=ds^2+ m_{\ell}^2g_{S^2},
    \]
    where $g_{S^2}$ is the round metric on $S^2$,
    one has
    \[
    \begin{split}
        &\bigl\|(\phi_\ell)_*g-g_{\ell}\bigr\|_
            {C^{2,\alpha}([s,\infty)\times S^2)}
        +\bigl\|(\phi_\ell)_*k\bigr\|_
            {C^{1,\alpha}([s,\infty)\times S^2)}
        \\
        &\quad
        +\bigl\|(\phi_\ell)_*E-\mathfrak{e}_\ell\partial_s\bigr\|_
            {C^{1,\alpha}([s,\infty)\times S^2)}
        +\bigl\|(\phi_\ell)_*B-\mathfrak{b}_\ell\partial_s\bigr\|_
            {C^{1,\alpha}([s,\infty)\times S^2)}
        \longrightarrow 0
    \end{split}
    \]
    as $s\to\infty$.

 \item The energy density $\mu$ and momentum density $J$, defined by
    \[
        \mu
        :=
        \frac{1}{2}
        \left(
            R_g-|k|_g^2+(\operatorname{tr}_g k)^2
        \right)
        -|E|_g^2-|B|_g^2
    \]
    and
    \[
        J
        :=
        \operatorname{div}_g
        \bigl(k-(\operatorname{tr}_g k)g\bigr)
        +2E\times_g B,
    \]
    belong to $L^1(M,g)$.
    
    \item The electric and magnetic fields are divergence-free\footnote{The PMT inequality still holds without this assumption by absorbing the divergence terms into the energy density. However, the rigidity Theorem \ref{T:main} becomes incorrect in this more general setting.}:
    \[
        \operatorname{div}_g E
        =
        \operatorname{div}_g B
        =
        0.
    \]
\end{enumerate}
\end{definition}
\begin{remark} \label{ebm}
Since $\mu\to m_\ell^{-2}-\mathfrak{e}_\ell^2-\mathfrak{b}_\ell^2$ along each cylinrical end,    $\mu\in L^1(M)$ implies that $\mathfrak{e}_\ell^2+\mathfrak{b}_\ell^2= m_\ell^{-2}$. 
\end{remark}
\begin{definition}
  The ADM energy $\mathcal{E}$ and the ADM linear momentum  $P$ are given by 
  \begin{equation*}
\mathcal{E}:=\frac{1}{16\pi}\lim_{r\to \infty}\int_{S_r}(g_{ij,i}-g_{ii,j})\nu^j dA_g,\qquad
 {P}_i:= \frac{1}{8\pi}\lim_{r\to \infty}\int_{S_r}(k_{ij}-(\operatorname{tr}_g k)g_{ij})\nu^j dA_g,
    \end{equation*}
    where $S_r:=\{|x|=r\}$  and $\nu$ is the outer unit normal to $S_r$.
The total
electric and magnetic charges of the end are defined by
\[
Q_E
 :=\frac{1}{4\pi}\lim_{r\to\infty}
 \int_{S_r}\langle E,\nu\rangle \,dA_g,
\qquad
Q_B
 :=\frac{1}{4\pi}\lim_{r\to\infty}
 \int_{S_r}\langle B,\nu\rangle \,dA_g.
\]
\end{definition}

Let $\mathcal{S}$ be the spinor bundle on $(M^3,g)$. Let $\overline{\mathcal S}=\mathcal{S}\oplus\mathcal{S}$ be the spacetime
spinor bundle.  We use
\[
 e_i e_j+e_j e_i=-2\delta_{ij},\qquad e_0^2=1,
 \qquad e_0e_i=-e_ie_0,\qquad \chi=e_1e_2e_3.
\]
 Let $\mathcal D:=e_i\nabla_i$ and set
\begin{equation*}
 \widehat\nabla_i
 :=\nabla_i+\frac12k_{ij}e_je_0
       -\frac12Ee_ie_0+\frac12Be_i\chi,
 \qquad
 \widehat{\mathcal D}:=e_i\widehat\nabla_i.
\end{equation*}
For $-q<\delta<-\frac12$, let $\mathcal X_\delta$ denote the hybrid
space which equals $W^{1,2}_\delta$ on each asymptotically flat end
and $H^1$ on each cylindrical end.  Similarly, let
$\mathcal Y_{\delta-1}$ equal $L^2_{\delta-1}$ on the asymptotically
flat ends and $L^2$ on the cylindrical ends.
\begin{proposition} \label{existence}
    Let $(M,g,k,E,B)$ be an admissible initial data set satisfying the charged dominant energy condition.  Let $\psi_\infty$ be an
 spinor equal to a unit constant spinor at infinity on the designated asymptotically flat end
and equal to zero near every other end.  Then there exists a spinor $\psi$ satisfying 
    \[\widehat{\mathcal{D}}\psi=0,\quad \psi-\psi_\infty\in \mathcal{X}_\delta.\]
\end{proposition}
\begin{proof}
 For $\phi\in C^2_c$, a direct computation yields
\begin{equation}\label{eq:compactcoercivity}
 \|\widehat{\mathcal D}\phi\|_2^2
 =\|\widehat\nabla\phi\|_2^2
 +\frac12\int_M\operatorname{Re}\langle
   (\mu+J^ie_ie_0)\phi,\phi\rangle\,d\mu_g
 \geq\|\widehat\nabla\phi\|_2^2.
\end{equation}
Thus,  $\widehat{\mathcal{D}}:\mathcal{X}_\delta\to \mathcal{Y}_{\delta-1}$ is injective. However, $\widehat{\mathcal{D}}$ is not self-adjoint, then we need to deform $\widehat{\mathcal{D}}$ to a self-adjoint operator. Let 
\[
 V:=
    -\frac12Ee_0+\frac12B\chi,
 \quad
 \overline{\mathcal{D}}=\mathcal{D}-\frac{1}{2}(\operatorname{tr}_g k)e_0
 \quad
 \mathcal D_t:=\overline{\mathcal D}+tV,\qquad 0\leq t\leq1.
\]
Thus $\mathcal D_0=\overline{\mathcal D}$ and
$\mathcal D_1=\widehat{\mathcal D}$. The operators
$
 \mathcal D_t:\mathcal X_\delta\longrightarrow
 \mathcal Y_{\delta-1}
$
form a continuous Fredholm family. On each asymptotically flat
end this follows from the decay of the coefficients. On a cylindrical
end, the limiting tangential operator is
\[
 A_{\ell,t}
 =
 A_{S^2}
 +t\left(
 -\frac{\mathfrak e_\ell}{2}e_0
 +\frac{\mathfrak b_\ell}{2}\chi
 \right),
\]
where $A_{S^2}$ is the Dirac operator of the spherical cross-section of the cylindrical end.
 Since the spectral gap
of $A_{S^2}$ is $m_\ell^{-1}$ and by Remark \ref{ebm}, we have
\[
 \operatorname{dist}
 \bigl(0,\operatorname{spec}A_{\ell,t}\bigr)
 \geq
 m_{\ell}^{-1}-\frac{t}{2}
 \sqrt{\mathfrak e_\ell^2+\mathfrak b_\ell^2}
 \geq\frac{1}{2m_{\ell}}.
\]
The Fredholm assertion follows by applying
\cite[Theorem~6.2]{LockhartMcOwen} on each asymptotically cylindrical
end and \cite[Theorem~9.1]{LockhartMcOwen} on each asymptotically
flat end, and then patching the resulting end estimates with an
interior elliptic estimate.
Note that the charged dominant energy condition implies the standard dominant energy condition: 
\[\frac{1}{2}
        \left(
            R_g-|k|_g^2+(\operatorname{tr}_g k)^2
        \right)\ge |\operatorname{div}_g(k-(\operatorname{tr}_gk)g)|_g.\]
        Hence, $\operatorname{ker} \overline{\mathcal{D}}=0$. Since $\overline{\mathcal{D}}$ is self-adjoint, we have $\operatorname{ind}\overline{\mathcal{D}}=0$, 
        therefore, $\operatorname{ind}\widehat{\mathcal{D}}=0$. 
        Combining with the injectivity of $\widehat{D}$,  $\operatorname{coker}\widehat{\mathcal{D}}=0$,i.e., $\widehat{\mathcal{D}}$ is an isomorphism. Since $\widehat{\mathcal{D}}\psi_\infty\in \mathcal{Y}_{\delta-1}$, there exists a spinor $\phi\in \mathcal{X}_\delta$ such that 
        $\widehat{\mathcal{D}}\phi=-\widehat{\mathcal{D}}\psi_\infty$.
\end{proof}
Next, we recall the integral formula for the positive mass theorem with electromagnetic fields; see \cite{BartnikChrusciel,GHHP,Witten}.
\begin{proposition}
For the spinor $\psi$ of Proposition~\ref{existence},
\begin{align}
 &\int_M\left(
 |\widehat\nabla\psi|^2-|\widehat{\mathcal D}\psi|^2
 +\frac12\operatorname{Re}\langle
    (\mu+J^ie_ie_0)\psi,\psi\rangle\right)d\mu_g \notag\\
 &\quad=4\pi\left\langle
 \bigl(\mathcal E+Pe_0+Q_Ee_0-Q_B\chi\bigr)
 \psi_\infty,\psi_\infty\right\rangle. \label{eq:globalwitten}
\end{align}
In particular, when $\mathcal{E}^2=|P|^2+Q_E^2+Q_B^2$, there exists $\psi_\infty$ such that 
\[\psi-\psi_\infty\in \mathcal{X}_\delta\quad \text{and}\quad \widehat{\nabla}\psi=0.\]
\end{proposition}

\begin{proposition} \label{X neq 0}
Suppose $\mathcal{E}^2=|P|^2+Q_E^2+Q_B^2$.
\begin{enumerate}
    \item If $\mathcal{E}=|P|$, then $E=B=0$ and $(M,g,k)$ embeds into Minkowski spacetime. 
    \item  If $\mathcal{E}>|P|$, then $N^2_\infty-|X_\infty|^2$ is a nonzero constant where 
    \[N_\infty=|\psi_\infty|^2,\quad (X_\infty)_i=\langle e_ie_0\psi_\infty,\psi_\infty\rangle.
    \]
\end{enumerate}
\end{proposition}
\begin{proof}
    1. In this case, the PMT with electromagnetic fields reduces to the spacetime PMT. Then there exists a spinor $\varphi$ satisfying $\nabla_i\varphi=-\frac{1}{2}k_{ij}e_je_0\varphi$, and $\varphi\to 0$ at cylindrical ends. However, an ODE argument shows that $\varphi=0$; see Lemma \ref{lem:oneend} for a similar, more detailed proof.
    Therefore, there is no cylindrical ends. Hence,  applying \cite{HirschZhang}, the result follows.

    2. Suppose $N_\infty=|X_\infty|$. Then 
    \[\langle e_0\psi_\infty,\psi_\infty \rangle=\langle\chi\psi_\infty,\psi_\infty\rangle=0.\]
    However, the boundary term in equation \eqref{eq:globalwitten} does not vanish. Thus, $N^2_\infty-|X_\infty|^2\ne0$. Moreover, since $\psi_\infty$ is a constant spinor, $N^2_\infty-|X_\infty|^2$ is a nonzero constant.
\end{proof}
From now on, we only consider the case $\mathcal{E}>|P|$.

\section{Topology of $M$}

\begin{lemma}\label{lem:oneend}
$M$ has exactly one asymptotically flat end.
\end{lemma}

\begin{proof} 
Suppose $\mathcal U'$ is another
asymptotically flat end.  Proposition~\ref{existence} gives a
solution with $\psi\to0$ on $\mathcal U'$, while it tends to
$\psi_\infty\neq0$ at the
distinguished end.  In an asymptotically flat spin frame,
the gradient equation $\widehat{\nabla}\psi=0$ takes the form
\[
 \partial_i\psi=-\sigma_i\psi,
 \qquad |\sigma_i|\leq Cr^{-1-q},
\]
where $\sigma_i$ is a linear operator and $|\sigma_i|$ is the pointwise operator norm. 
Along an outward radial ray $\gamma$ toward $\mathcal{U}'$, Gronwall's inequality gives
\[
 |\psi(\gamma(r))|\geq |\psi(p)|\exp\left(-\frac{C}{q r(p)^q}\right).
\]
Letting $r\to\infty$ and using $\psi(\gamma(r))\to0$ yields $\psi(p)=0$.
The gradient equation implies $\psi\equiv0$ on $M$, a
contradiction.  Hence, no second asymptotically flat end exists.
\end{proof}

\begin{lemma}
\label{lem:simplyconnected}
$M$ is simply connected.  In particular,
$H^1(M;\R)=0$. 
\end{lemma}

\begin{proof}
    If $M$ is not simply connected, then by passing to a finite cover, we obtain a manifold $\widetilde{M}$ that possesses multiple asymptotically flat ends and saturates the mass inequality. However, this contradicts Lemma \ref{lem:oneend}.
\end{proof}

\section{Existence and properties of the Tod map}

Let $\psi=(\psi_1,\psi_2)\in\overline{\mathcal S}(M)$ where $M$ is a $3$-dimensional manifold and $\overline{\mathcal S}=\mathcal S\oplus\mathcal S$ is the spacetime spinor bundle.
Define the vector fields 
\begin{align*}
    X_i=&\langle e_ie_0\psi,\psi\rangle=-2\operatorname{Re}\langle e_i\psi_1,\psi_2\rangle,\\
    Y_i=& \operatorname{Im}\langle e_i \chi\psi,\psi\rangle=\operatorname{Im}\langle e_i\psi_1,\psi_1\rangle+\operatorname{Im}\langle e_i\psi_2,\psi_2\rangle,\\
    \widetilde Y_i=& \operatorname{Im}\langle e_i \mathcal J\chi\psi,\psi\rangle=\operatorname{Im}\langle e_i\mathcal J\psi_1,\psi_1\rangle+\operatorname{Im}\langle e_i\mathcal J\psi_2,\psi_2\rangle,\\
    \widehat Y_i=& \operatorname{Re}\langle e_i \mathcal J\chi\psi,\psi\rangle=\operatorname{Re}\langle e_i\mathcal J\psi_1,\psi_1\rangle+\operatorname{Re}\langle e_i\mathcal J\psi_2,\psi_2\rangle,
\end{align*}
and the functions
\begin{align*}
N=&|\psi|^2=|\psi_1|^2+|\psi_2|^2,\\
        R=& \operatorname{Im} \langle e_0\chi \psi,\psi\rangle =2\operatorname{Im}\langle \psi_1,\psi_2\rangle,\\
               \widetilde R=& \operatorname{Im} \langle e_0\mathcal J\chi \psi,\psi\rangle= 2\operatorname{Im}\langle \mathcal J \psi_1,\psi_2\rangle,\\
                     \widehat  R=& \operatorname{Re} \langle e_0\mathcal J\chi \psi,\psi\rangle =2\operatorname{Re}\langle \mathcal J \psi_1,\psi_2\rangle.
\end{align*}
Here $\chi=e_1e_2e_3$ is the chirality operator, and $\mathcal J$ is the quaternionic structure on $\mathcal S$, cf. \cite[Remark 2.13]{AHM} and \cite[p.~32]{Friedrich}. It commutes with Clifford multiplication and differentiation and satisfies $\mathcal J^2=-1$ as well as $\mathbf i\mathcal J=-\mathcal J\mathbf i$.
We also define $\mathbf X=(N,X),\mathbf Y=(R,Y),\widetilde{\mathbf Y}=(\widetilde R,\widetilde Y), \widehat{\mathbf Y}=(\widehat R,\widehat Y)$ and note that $\mathbf X,\mathbf Y, \widetilde{\mathbf Y},\widehat{\mathbf Y}$ can be regarded as vector fields in the Killing development of $(M,g,k,N,X)$.
Here the Killing development $(\overline M,\overline g)$ is defined by $\overline M=M\times\mathbb R$ and $\overline g=-N^2dt^2+g_{ij}(dx^i+X^idt)(dx^j+X^jdt)$.
Finally, we set
\begin{align*}
    L=&\langle e_0\psi,\psi\rangle=2\operatorname{Re}\langle\psi_1,\psi_2\rangle,\\
    K=&\langle \chi \psi,\psi\rangle=|\psi_1|^2-|\psi_2|^2.
\end{align*}

\begin{theorem}
    We have
    \begin{align*}
        N^2-|X|^2=|Y|^2-R^2=|\widetilde Y|^2-\widetilde R^2=|\widehat Y|^2-\widehat R^2=K^2+L^2.
    \end{align*}
    Moreover,
    \begin{align*}
        0=&\langle X,Y\rangle -NR=\langle X,\widetilde Y\rangle-N\widetilde R=\langle X,\widehat Y\rangle -N\widehat R,\\
        0=&\langle Y,\widetilde Y\rangle -R\widetilde R=\langle Y,\widehat Y\rangle-R\widehat R=\langle \widetilde Y,\widehat Y\rangle-\widetilde R\widehat R.
    \end{align*}
    In other words, the vector fields $\mathbf X,\mathbf Y,\widetilde{\mathbf Y}, \widehat{\mathbf Y}$ are pairwise perpendicular in the Killing development and have the same absolute Lorentzian norm.
\end{theorem}

Note that in case $\psi_2=0$, this reduces to $\operatorname{Im}\langle e_i\psi_1,\psi_1\rangle$, $\operatorname{Im}\langle \mathcal J e_i\psi_1,\psi_1\rangle$, $\operatorname{Re}\langle\mathcal J e_i\psi_1,\psi_1\rangle$ forming three perpendicular vector fields of the same length on $M$.

\begin{proof}
    In dimension 3, we can identify a spinor $\psi_1\in \mathcal S$ with a pair of complex numbers $(a,b)\in\mathbb C^2$.
The Pauli matrices are given by
\begin{align*}
\sigma_1= \begin{pmatrix}
0 & 1\\
1 & 0
\end{pmatrix},\quad \sigma_2= \begin{pmatrix}
0 & -i\\
i & 0
\end{pmatrix},\quad \sigma_3= \begin{pmatrix}
1 & 0\\
0 & -1
\end{pmatrix},
\end{align*}
and the unit vectors $e_1,e_2,e_3$ act on $\mathcal S$ via $i\sigma_1,i\sigma_2,i\sigma_3$.
Moreover, $\mathcal J\psi_1=e_2\overline\psi_1$ where $\overline\psi_1$ denotes complex conjugation.
Using this setup, the above result is a direct computation.
\end{proof}

Next, we define 
\begin{align*}
    W_i=&\operatorname{Im}\langle e_i\psi_1,\psi_1\rangle-\operatorname{Im}\langle e_i\psi_2,\psi_2\rangle,\\
    Z_i=&-2\operatorname{Im}\langle e_i\psi_1,\psi_2\rangle,\\
    \widetilde    W_i=&\operatorname{Im}\langle \mathcal J e_i\psi_1,\psi_1\rangle-\operatorname{Im}\langle \mathcal J e_i\psi_2,\psi_2\rangle,\\
 \widetilde   Z_i=&-2\operatorname{Im}\langle\mathcal J  e_i\psi_1,\psi_2\rangle,\\
\widehat        W_i=&\operatorname{Re}\langle\mathcal J  e_i\psi_1,\psi_1\rangle-\operatorname{Re}\langle\mathcal J  e_i\psi_2,\psi_2\rangle,\\
  \widehat   Z_i=&-2\operatorname{Re}\langle \mathcal J e_i\psi_1,\psi_2\rangle.
\end{align*}
Note that $W,Z,\widetilde W,\widetilde Z,\widehat W,\widehat Z$ are not spacetime quantities, i.e. they depend on the choice of $e_0$.
However, they will always be paired with $E,B$ which also depend on $e_0$.
Like the Faraday tensor $F$, the pairs $(W,Z)$, $(\widetilde W,\widetilde Z)$, $(\widehat W,\widehat Z)$ correspond to a $2$-form in the spacetime.

\begin{theorem}
Suppose that $\psi$ is super-covariantly constant, i.e.
\begin{align*}
         \nabla_i \psi=\left(-\frac{1}{2}k_{ij}e_je_0+\frac{1}{2}Ee_ie_0-\frac{1}{2}B e_i\chi\right)\psi.
\end{align*}
Then
\begin{align*}
    \nabla_iN=&-k_{ij} X_j - LE_i+KB_i,
  \\  \nabla_iX_j=&-k_{ij}N+ K(\ast E)_{ij}+L(\ast B)_{ij},
   \\ \nabla_iR=&-k_{ij}Y_j+(\ast E)_{ij}W_j-(\ast B)_{ij}Z_j,
  \\  \nabla_iY_j=&  -k_{ij}R    +Z_iE_j+Z_jE_i-\langle Z,E\rangle g_{ij}+W_iB_j+W_jB_i-\langle W,B\rangle g_{ij}  ,
   \\ \nabla_i\widetilde   R=&-k_{ij}\widetilde   Y_j+(\ast E)_{ij}\widetilde   W_j-(\ast B)_{ij}\widetilde   Z_j,
  \\  \nabla_i\widetilde   Y_j=&  -k_{ij}\widetilde   R    +\widetilde   Z_iE_j+\widetilde   Z_jE_i-\langle\widetilde    Z,E\rangle g_{ij}+\widetilde   W_iB_j+\widetilde   W_jB_i-\langle \widetilde   W,B\rangle g_{ij}  ,
     \\ \nabla_i\widehat  R=&-k_{ij}\widehat  Y_j+(\ast E)_{ij}\widehat  W_j-(\ast B)_{ij}\widehat  Z_j,
  \\  \nabla_i\widehat  Y_j=&  -k_{ij}\widehat R    +\widehat  Z_iE_j+\widehat  Z_jE_i-\langle \widehat  Z,E\rangle g_{ij}+\widehat  W_iB_j+\widehat  W_jB_i-\langle \widehat  W,B\rangle g_{ij}  .
\end{align*}
In particular, $dY^\flat=d\widetilde Y^\flat=d\widehat Y^\flat=0$ and $(N,X)$ is a Killing vector field in the Killing development of $(M,g,k)$.
\end{theorem}

\begin{proof}
    This is a direct computation.
\end{proof}

\begin{definition}
    We call $T=(x_1,x_2,x_3):M\to\mathbb R^3$ Tod map if 
    \begin{align*}
        DT=(Y,\widetilde Y,\widehat Y).
    \end{align*}
\end{definition}

Note that $T$ always exists since $M$ is simply connected by Lemma \ref{lem:simplyconnected}.

\begin{lemma}\label{T}
Let $v:=\sqrt{N^2-|X|^2}$. Then
    we have
    \begin{align*}
        T^\ast \delta =v^2g+X^\flat\otimes X^\flat
    \end{align*}
    and
    \begin{align*}
        dx_1\wedge dx_2\wedge dx_3
 &=Nv^2d\mu_g.
    \end{align*}
    In particular, the singular values of the Tod map are $v, v, N$ and the Tod map is orientation preserving.
\end{lemma}

\begin{proof}
    This is another direct computation.
\end{proof}

\section{Mapping degree of the Tod map and compactness of the null set}

In this section, we prove that the Tod map $T$ has degree one by analyzing its behavior in the asymptotic regions. 
More precisely, we show that $T$ is proper away from the cylindrical ends and has degree one in the asymptotically flat end.
Since $T\to x_a$ in each cylindrical end $\mathcal C_a$ for some $x_a\in\mathbb R^3$, the result follows.

\medskip

Additionally, we prove that the null set where $N=|X|$ is contained away from both the asymptotically flat and the cylindrical ends.

\subsection{Tod map in the asymptotically flat end}

\begin{lemma}
    We have
\begin{equation*}
 DT=\mathcal{L}+O(r^{-q}),
\end{equation*}
and
\begin{equation*}
   T(y)=\mathcal{L}y+o(r) 
\end{equation*}
for some $\mathcal{L}\in GL^+(3,\R)$.
In particular, the null set where $N=|X|$ is contained away from the asymptotically flat end.
\end{lemma}

\begin{proof}
    Let $y=(y^1,y^2,y^3)$ denote the asymptotically Euclidean coordinate on
the asymptotically flat end and put $r=|y|$.  
$\psi$ satisfies
\[
 \psi=\psi_\infty+O_2(r^{-q}).
\]
Indeed, in an asymptotically flat spin frame
$\partial_i\psi=\sigma_i\psi$ with
$\sigma_i=O(r^{-1-q})$; 
boundedness and radial integration give the estimate.  Proposition \ref{X neq 0} shows  that $|\mathbf{X}_\infty|$ is a nonzero constant.    Let $\mathcal{L}$ be the limiting derivative of the
Tod map in the coordinate system $y$.  Lemma \ref{T} gives
\begin{equation}\label{eq:LGram}
 \mathcal{L}^{\mathsf T}\mathcal{L}=|\mathbf{X}_\infty|^2I+X_\infty\otimes X_\infty,
 \qquad \det \mathcal{L}=N_\infty|\mathbf{X}_\infty|^2>0.
\end{equation}
Thus $\mathcal{L}\in GL^+(3,\R)$ and
\begin{equation}\label{eq:DxAF}
 DT=\mathcal{L}+O(r^{-q}).
\end{equation}
We now integrate this estimate carefully.  Define
\[
 \mathring{T}(y):=T(y)-\mathcal{L}y.
\]
Then $|D\mathring{T}|\leq Cr^{-q}$.  Fix a large $R_0$ and write
$y=r\omega$, $\omega\in S^2$.  Integration along the radial segment gives,
uniformly in $\omega$,
\[
 \mathring{T}(r\omega)=\mathring{T}(R_0\omega)
 +\int_{R_0}^{r}\partial_s\mathring{T}(s\omega)ds,
\]
and hence
\[
 |\mathring{T}(r\omega)|\leq C_0+C\int_{R_0}^{r}s^{-q}\,ds.
\]
It follows that
$
 \mathring{T}(y)=
 o(r).$
Therefore, we have 
\begin{equation} \label{eq:xAF}
   T(y)=\mathcal{L}y+o(r) 
\end{equation}
\end{proof}

\begin{lemma}
    Fix a point $p\in\mathbb R^3$ with $|p|\gg1$. 
    Then $p$ is a regular value and
    \begin{align*}
        \deg(T,p)=\sum_{x\in T^{-1}(p)}\operatorname{sgn} \det DT(x)=1.
    \end{align*}
\end{lemma}

\begin{proof}
Choose $R_0$ sufficiently large so that on the asymptotically flat region
$\{r\ge R_0\}$ we have
\[
    \det DT>0.
\]
This is possible since
\[
    DT=\mathcal{L}+O(r^{-q})
\]
and $\det \mathcal{L}>0$.

Moreover, $T$ is bounded on the complement of $\{r\ge R_0\}$.
Indeed, the interior region has bounded image, while on every
cylindrical end the map $T$ converges to a finite point, as will be proved independently in Lemma~\ref{lem:cylinderlimit}. Hence, there exists $C_0>0$ such that
\[
    |T(x)|\le C_0
    \qquad\text{for every }x\notin\{r\ge R_0\}.
\]
Thus, if $|p|>C_0$, every point of $T^{-1}(p)$ lies in the
asymptotically flat region $\{r>R_0\}$.

It follows that $p$ is a regular value. Indeed, if $x\in T^{-1}(p)$,
then $r(x)>R_0$, and therefore
\[
    \det DT(x)>0.
\]
In particular, every preimage of $p$ has local degree $+1$.

We next compute the degree. Since
\[
    T(y)=\mathcal{L}y+o(r),
\]
we have $|T(y)|\to\infty$ as $r\to\infty$. Hence, for $R$ sufficiently
large,
\[
    T^{-1}(p)\subset \{R_0<r<R\}.
\]
Set
\[
    A_{R_0,R}:=\{R_0<r<R\}.
\]
Since $p\notin T(\partial A_{R_0,R})$, the degree
$\deg(T,A_{R_0,R},p)$ is well-defined.

On the outer boundary $S_R$, consider
\[
    \Phi_{R,p}(\omega)
    :=
    \frac{T(R\omega)-p}{|T(R\omega)-p|}.
\]
Using $T(R\omega)=R\mathcal{L}\omega+o(R)$, we have
\[
    T(R\omega)-p
    =
    R\mathcal{L}\omega+o(R),
\]
uniformly in $\omega\in S^2$. Since $\mathcal{L}$ is invertible, for $R$
sufficiently large the straight-line homotopy between
$T(R\omega)-p$ and $R\mathcal{L}\omega$ does not pass through the origin.
Consequently,
\[
    \deg(\Phi_{R,p})
    =
    \deg\left(
        \omega\mapsto\frac{\mathcal{L}\omega}{|\mathcal{L}\omega|}
    \right)
    =
    \operatorname{sgn}\det \mathcal{L}
    =
    1.
\]

On the inner boundary $S_{R_0}$, choose $|p|$ still larger if
necessary so that
\[
    |p|>\sup_{S_{R_0}}|T|.
\]
Then
\[
    \frac{T-p}{|T-p|}:S_{R_0}\to S^2
\]
is homotopic to the constant map $-p/|p|$, and therefore has degree
zero.

By the boundary characterization of the Brouwer degree,
\[
    \deg(T,A_{R_0,R},p)=1.
\]
Since $A_{R_0,R}$ contains all preimages of $p$, this is precisely
$\deg(T,p)$. Thus
\[
    \deg(T,p)
    =
    \sum_{x\in T^{-1}(p)}
    \operatorname{sgn}\det DT(x)
    =
    1.
\]
\end{proof}

\subsection{Tod map at cylindrical ends}

\begin{lemma}
\label{lem:cylinderlimit}
On every cylindrical end there is a point
$x_\ell\in\R^3$ such that
\begin{equation*}
 \sup_{(s,\omega)\in[t_0,\infty)\times S^2}|T(s,\omega)-x_\ell|
 \longrightarrow0
\end{equation*}
for $t_0\to\infty$.
\end{lemma}

\begin{proof}
Write $S^2_{\ell,s}=\{s\}\times S^2$ and let $g_\ell(s)$ be its
induced metric.  Put
\[
 n_\ell(s)=\sup_{b\in S^2}N(s,b),
 \qquad N=|\psi|^2.
\]
  Note that $\widehat{\nabla}\psi=0$ implies
$|\nabla \psi|\leq C|\psi|$.  Since $\psi$ never vanishes, integration of
$|\nabla\log|\psi||\leq C$ along cross-sectional curves, together with the
diameter and volume bounds, gives a uniform comparison estimate
\begin{equation}\label{eq:crosscompare}
 n_\ell(s)\leq C\int_{S^2_{\ell,s}}N\,dA_{g_\ell(s)}.
\end{equation}
Proposition~\ref{existence} gives
$\psi\in L^2(\mathcal C_a)$ on the cylindrical end, therefore, Fubini implies
\begin{equation}\label{eq:nintegrable}
 \int_0^\infty n_\ell(s)\,ds<\infty.
\end{equation}
The same parallel-transport comparison in the longitudinal direction
shows, for $s$ large,
\begin{equation}\label{eq:nlocalaverage}
 n_\ell(s)\leq C\int_{s-1}^{s+1}n_\ell(t)\,dt.
\end{equation}
Hence $n_\ell(s)\to0$.
For a unit vector $Z$, Lemma \ref{T} gives
\begin{equation}\label{eq:Dxbound}
 |DT(Z)|^2=v^2+\langle X,Z\rangle^2
 \leq v^2+|X|^2=N^2.
\end{equation}
Thus $\|DT\|_{\rm op}\leq N$.  Fix $w_0\in S^2$.  Along the
longitudinal curve,
\[
 |T(s_2,w_0)-T(s_1,w_0)|
 \leq C\int_{s_1}^{s_2}n_\ell(s)\,ds.
\]
Equation \eqref{eq:nintegrable} makes $T(s,w_0)$ Cauchy.
We denote its limit $x_\ell$.  A cross-sectional path of uniformly bounded length gives
\[
 \sup_{w\in S^2}|T(s,w)-T(s,w_0)|\leq Cn_\ell(s).
\]
Finally, \eqref{eq:nlocalaverage} bounds the supremum of the last term on
the entire tail by a constant times
$\int_{t_0-1}^\infty n_\ell ds$.  This completes the proof.
\end{proof}

In the next lemma, we will show that the null set $\{N=|X|\}$ stays away from the cylindrical ends.
See also Theorem 3.2 in \cite{CRT}.

\begin{lemma} \label{lem:cylindrical}
    On each cylindrical end,  
    \[\frac{|X|}{N}\to 0,\quad \text{as}\quad s\to \infty.\]
\end{lemma}

\begin{proof}
Let 
\[Q:=\frac{1}{2}(\mathfrak{e}_\ell e_0-\mathfrak{b}_\ell\chi),\qquad \mathfrak{q}:=\frac{1}{2}(\mathfrak{e}_\ell^2+\mathfrak{b}^2_\ell)^\frac{1}{2}.
\]
Then  $Q^*=Q$, $Q^2=\mathfrak q^2I$.
    The gradient equation $\widehat{\nabla}\psi=0$ and the decay assumption in definition \ref{admissible} imply 
    \begin{align}
       \partial_s\psi+Q\psi
   &=\mathcal R\psi, \label{s ODE}                        \\
 \nabla_A^{\sigma}\psi-e_se_AQ\psi
   &=\mathcal R_A\psi, \label{A ODE}
   \end{align}         
   where $\mathcal R$   and $\mathcal R_A$ are remainder terms satisfying
   \begin{equation*}
        \lim_{s\to \infty}\sup_{\omega\in S^2}
 \bigl(|\mathcal R|+|\mathcal R_A|\bigr)=0. 
   \end{equation*}
Then equation \eqref{A ODE} implies that there exists $C>0$ such that
\[\sup_{\omega\in S^2}|\psi(s,\omega)|\le C\inf_{\omega\in S^2}|\psi(s,\omega)|.\]
Since $\psi$ is $L^2$ integrable in the cylindrical end, we have 
\begin{equation} \label{L2s}
  \int_0^\infty |\psi(s,\omega)|^2ds<\infty.  
\end{equation}
Let $Q_\pm:=\frac{1}{2}(I\pm \mathfrak{q}^{-1}Q)$, $\psi_\pm:=Q_\pm \psi$. Applying $Q_\pm$ to equation \eqref{s ODE} and using $Q\psi_\pm=\pm\mathfrak{q}\psi_\pm$, we have
\begin{equation}\label{pm ODE}
\partial_s \psi_\pm \pm\mathfrak{q}\psi_\pm= Q_\pm\mathcal{R}\psi.
\end{equation}

Fix $\omega\in S^2$. Denote $\upsilon_+(s)=|\psi_+(s,\omega)|^2$ and $\upsilon_-(s)=|\psi_-(s,\omega)|^2$. For $s>0$, define $\varepsilon(s):=\sup_{t\ge s}\|\mathcal{R}\|$, then $\varepsilon(s)\to 0$. From equation \eqref{pm ODE}, we have 
\begin{equation} \label{uv ode}
    |\upsilon_+'+2\mathfrak{q}\upsilon_+|\le 4\varepsilon(s)(\upsilon_++\upsilon_-),\quad |\upsilon_-'-2\mathfrak{q}\upsilon_-|\le 4\varepsilon(s)(\upsilon_++\upsilon_-).
\end{equation}
Let $\kappa:=\upsilon_+-\upsilon_-$ and $\zeta:=\upsilon_+ +\upsilon_- $, then
\begin{equation}\label{eq:D-inequality}
    \kappa'\le -2(\mathfrak{q}-4\varepsilon(s))\zeta.
\end{equation}
Because $\zeta=\upsilon_+ +\upsilon_-=|\psi(s)|^2\in L^1(0,\infty)$, there is a sequence $T_j\to \infty$ such that $\zeta(T_j)\to 0$, and hence $\kappa(T_j)\to 0$. Choose \(s_0>0\) such that \(\varepsilon(s_0)<\frac{1}{4}\mathfrak{q}\). For \(s\ge s_0\),
integrating \eqref{eq:D-inequality} from \(s\) to \(T_j\), and using the fact that $\varepsilon$ is nonincreasing  gives
\[
    \kappa(T_j)-\kappa(s)
    \le
    -2\bigl(\mathfrak{q}-4\varepsilon(s)\bigr)
    \int_s^{T_j}\zeta(t)\,dt.
\]
Letting \(j\to\infty\), we obtain
\begin{equation} 
    \kappa(s)
    \ge
    2\bigl(\mathfrak{q}-4\varepsilon(s)\bigr)
    \int_s^\infty \zeta(t)\,dt.
    \label{eq:D-positive}
\end{equation}
The integral on the right-hand side is strictly positive: otherwise
\(\zeta\equiv0\) on \([s,\infty)\), contradicting the assumption
\(|\psi|\neq0\). Thus
\[
    \kappa(s)>0,
    \qquad s\ge s_0.
\]
Consequently,
\[
    0\le\upsilon_-(s)<\upsilon_+(s),
    \qquad s\ge s_0.
\]

From the second equation in \eqref{uv ode},
we obtain
\[
    \upsilon_-'-2\mathfrak q\upsilon_-
    \ge -4\varepsilon(s)\zeta.
\]
Consequently,
\[
    \left(e^{-2\mathfrak q s}\upsilon_-(s)\right)'
    \ge
    -4e^{-2\mathfrak q s}\varepsilon(s)\zeta.
\]
Since $\zeta(T_j)\to 0$, we have \(\upsilon_-(T_j)\to0\). Integrating the preceding inequality
from \(s\) to \(T_j\), and then letting \(j\to\infty\), gives
\begin{align}
    \upsilon_-(s)
    &\le
    4\int_s^\infty
    e^{-2\mathfrak q(t-s)}
    \varepsilon(t)\zeta\,dt \notag\\
    &\le
    4\varepsilon(s)
    \int_s^\infty \zeta\,dt,
    \label{eq:v-tail-estimate}
\end{align}
where we used the fact that \(\varepsilon\) is nonincreasing.

We may therefore define
\[
    \eta(s):=\frac{\upsilon_-(s)}{\upsilon_+(s)}.
\]
Combining \eqref{eq:v-tail-estimate} and
\eqref{eq:D-positive}, and using \(\upsilon_+\ge\kappa\), we obtain
\[
    0\le\eta(s)
    \le
    \frac{2\varepsilon(s)}
         {\mathfrak q-4\varepsilon(s)}.
\]
Since \(\varepsilon(s)\to0\), it follows that
\[
    \lim_{s\to\infty}\eta(s)=0,
\qquad \text{i.e.}\quad
    \frac{|\psi_-(s,\omega)|}
         {|\psi_+(s,\omega)|}
    =\sqrt{\eta(s)}
    \longrightarrow0.
\]

Finally, since \(Q\) anticommutes with \(Y_0e_0\) for every unit
tangent vector \(Y_0\),
\[
    \langle X,Y_0\rangle
    =
    2\operatorname{Re}
    \left\langle Y_0e_0\psi_+,\psi_-\right\rangle.
\]
Therefore,
\[
    |X|\le2|\psi_+||\psi_-|=2\sqrt{\upsilon_+\upsilon_-}.
\]
Since \(N=|\psi|^2=\upsilon_+ +\upsilon_-\), we conclude that
\[
    \frac{|X|}{N}
    \le
    \frac{2\sqrt{\eta}}{1+\eta}
    \longrightarrow0
    \qquad\text{as }s\to\infty.
\]
\end{proof}

\subsection{Degree computation} Let
\[
    \mathcal P:=\{x_a:1\le a\le \ell\}\subset\mathbb R^3.
\]
For $y\in \mathbb{R}^3\setminus\mathcal{P}$, we denote 
$\deg_{\mathcal P}(T,y)=\deg (T|_{M\setminus T^{-1}(\mathcal P)},y)$. 
\begin{proposition}
\label{prop:degreeone}
For every compact set
$\mathcal C\subset\mathbb R^3\setminus\mathcal P$,
the set $T^{-1}(\mathcal C)$ is compact. Moreover,
\[
    \deg_{\mathcal P}(T,y)=1
    \qquad\text{for every }
    y\in\mathbb R^3\setminus\mathcal P.
\]
In particular, if $y$ is a regular value, then $T^{-1}(y)$
consists of exactly one point.
\end{proposition}

\begin{proof}
The asymptotic relation
\[
    T(x)=\mathcal{L}x+o(|x|),\qquad L\in GL^+(3,\mathbb R),
\]
implies that $|T(x)|\to\infty$ along the asymptotically flat end.
On the $a$-th cylindrical end, $T$ converges uniformly by Lemma \ref{lem:cylinderlimit}. Therefore,
$
 T:M\setminus T^{-1}(\mathcal P)
   \longrightarrow\mathbb R^3\setminus\mathcal P
$
is proper.

By the preceding asymptotically flat degree computation, there exists
$y_0\in\mathbb R^3\setminus\mathcal P$ such that
$
    \deg_{\mathcal P}(T,y_0)=1.$ 
Thus, the invariance
of the degree of a proper map gives 
$
    \deg_{\mathcal P}(T,y)=1$,
for every $y\in\mathbb R^3\setminus\mathcal P$,  see
\cite[p.~99]{OutereloRuiz}.

If $y$ is a regular value, then $T^{-1}(y)$ is finite. Furthermore,
\[
    \det DT=Nv^2>0
\]
at every point of $T^{-1}(y)$, so every preimage has local degree
$+1$. Hence
\[
    \#T^{-1}(y)=\deg_{\mathcal P}(T,y)=1.
\]
\end{proof}

\section{Structure of the null set}
Set
\[
    \mathcal V:=L+\mathrm iK,
    \qquad
    v:=|\mathcal V|
          =\sqrt{N^2-|X|^2}
          , 
           \qquad   Z_M:=\{v=0\}. 
\]
    The goal of this section is to prove that $Z_M$ is empty. 
    
     To the contrary, suppose that $Z_M$ is nonempty.
Proposition~\ref{X neq 0}
 shows that $v>0$ sufficiently far
out on the asymptotically flat end. Moreover, 
since $N \neq 0$ everywhere and Lemma~\ref{lem:cylindrical} gives 
$v\neq 0$ at the cylindrical end. Consequently, $Z_M$ is compact.

Here is the definition of the relative degree which we will use later, see \cite[p.~145]{OutereloRuiz}.
\begin{definition} Let $\Omega\subset M$ be an open subset. For a regular value $x\notin T(\partial\Omega)$, the relative degree
$\deg(T,\Omega,x)\in\mathbb Z$ is defined by
    \[
    \deg(T,\Omega,x)
    =
    \sum_{p\in T^{-1}(x)\cap\Omega}
    \operatorname{sgn}\det(DT)_p.
\]
For a critical value $x$, the degree is defined using a nearby regular value in the same component of $\mathbb{R}^3\setminus T(\partial\Omega)$.
\end{definition}

\subsection{Local structure}

\begin{theorem}\label{local plaque}
For every $p\in Z_M$, there exist a neighborhood $U$ of $p$ and a $C^2$
embedded disk $\Sigma_0\subset U$ such that
\[
    p\in\Sigma_0
    \subset Z_M\cap T^{-1}(T(p)).
\]
\end{theorem}

\begin{proof}
    At every interior point of $Z_M$, Lemma~\ref{T} shows that the singular values of
$DT$ are $0,0,N$. Since $N>0$, the rank of $DT$ is one there.  Suppose that $p\in\operatorname{Int}_M Z_M$. After restricting to a
neighborhood contained in $Z_M$, the constant-rank theorem gives a smooth
two-dimensional disk through $p$ on which $T$ is constant.

It remains to consider $p\in\partial Z_M$. 
In a small neighborhood of $p$, define
\[
 e:=\frac{X}{|X|},
 \qquad
 u:=\frac{DT(e)}{N}.
\]
The identity in Lemma~\ref{T} implies that $|u|=1$, and that $DT$ maps $e^\perp$ conformally with factor $v$ into $u^\perp$. 
Fix the constant target vector
$
    \ell:=u(p)$
and define
$$
    s=\ell\cdot T,
    \qquad
    w=P_{\ell^\perp}T.
$$
After translating $T$ in the target, we may assume that
\[
    T(p)=0,
    \qquad s(p)=0.
\]
At $p$,
\[
    ds=N e^\flat\neq0,
    \qquad
    dw=0.
\]
Hence, $s$ is a submersion around $p$.
Choose a local product chart
\[
    \Phi:(-\varepsilon,\varepsilon)\times D\longrightarrow U\subset M,
    \qquad
    \Phi(0,0)=p,
    \qquad
    s(\Phi(t,z))=t,
\]
where $D\subset\R^2$ is a disc, and identify
$
    \mathbb R^3=\mathbb R\ell\oplus\ell^\perp
$.
Thus, for each fixed $s$, the map $z\mapsto\Phi(s,z)$
parametrizes the level surface
\[
    \Sigma_t:=\{q\in U:s(q)=t\}.
\]

Set $w_t:=w|_{\Sigma_t}$.
The rest of the proof is lengthy, we break into several claims.

\textbf{Claim 1.}
After shrinking $U$, the map $w_0$ is either constant or quasiregular
satisfying \eqref{eq:quasiregular}.

Let $A:=DT$ and $c:=\langle\ell,u\rangle$. Since $c(p)=1$, shrinking the neighborhood
if necessary, we may assume $c\ge \frac{2}{3}$.  Fix
$p'\in\Sigma_0$ such that $v(p')\neq 0$. Let $\tau\in T_{p'}\Sigma_0$,  write
\[
    \tau=\alpha e+\xi,
    \qquad
    \xi\perp e.
\]
 Write
$A\xi=v\mathcal R_{p'}\xi$ for an isometry
$\mathcal R_{p'}:e^\perp\to u^\perp$. The level-set condition gives
\[
    0=\ell\cdot A\tau
      =\alpha Nc+v\langle\ell,\mathcal R_{p'}\xi\rangle.
\]
Since the component of $\ell$ perpendicular to $u$ has length
$\sqrt{1-c^2}$,
\[
    \alpha^2N^2c^2=v^2\langle \ell, \mathcal{R}_{p'}\xi\rangle^2
    \le v^2(1-c^2)|\xi|^2.
\]
Moreover,
\[
    |A\tau|^2=\alpha^2N^2+v^2|\xi|^2,
    \qquad
    |A\tau|^2-v^2|\tau|^2
       =\alpha^2|X|^2\ge0.
\]
Since $\tau$ is perpendicular to $\nabla s$, we have 
$A\tau\perp \ell$; therefore, $Dw_0(\tau)=A\tau$. 
It follows that,
at every point where $v>0$,
\begin{equation}\label{eq:w-distortion}
    v|\tau|
    \le |Dw_0(\tau)|
    \le \frac{v}{c}|\tau|
    \le \frac{3}{2}v|\tau|.
\end{equation}
 At a point where $v=0$, then $Dw_0=0$. Hence, there exists $1\le \mathcal{K}\le \frac{3}{2}$ such that
\begin{equation} \label{eq:quasiregular}
  \sigma_{\max}(Dw_0)\le 
  \mathcal{K}\sigma_{\min}(Dw_0).  
\end{equation}

\textbf{Claim 2.}
If $w_0$ is nonconstant, then its local index at $p$ satisfies
$i(p,w_0)\ge2$.

Suppose that $w_0$ is nonconstant. Note that the boundedness of $\mathcal{K}$ implies  that $w_0$ is quasiregular.  By the Stoilow factorization theorem \cite[p.179]{AIM}, locally $w_0=\varphi\circ h$, where $h$ is
a $\mathcal{K}$-quasiconformal homeomorphism and $\varphi$ is holomorphic. Moreover, 
$w_0$ is open and discrete, and its positive local index
$i(p,w_0)$ is a positive integer equal to the vanishing order of the holomorphic
function
$\varphi-\varphi(h(p))$ at $h(p)$; see \cite[p.180]{AIM}.
Thus,
if $i(p,w_0)=1$,  then
$\varphi'(h(p))\neq0$. The  inverse function theorem shows that $\varphi$, and hence $w_0$, is locally invertible.
Consequently, $w_0$ is a local $\mathcal{K}$-quasiconformal homeomorphism. Mori's estimate \cite[pp.~81--82]{AIM} applied to the
inverse gives, in a local coordinate $z$ centered at $p$,
\[
    |w_0(z)-w_0(0)|\ge C^{-1}|z|^{\mathcal{K}}.
\]
On the other hand, $Dw_0(p)=0$ and smoothness give
\[
    |w_0(z)-w_0(0)|\le C|z|^2.
\]
However, we have a contradiction as $\mathcal{K}<\frac{3}{2}$. Hence
$i(p,w_0)\ge2$.

\textbf{Claim 3.} $w_0$ is constant.

Suppose otherwise. By Claim~2, $i(p,w_0)\ge 2$.
Since $w_0$ is discrete, choose a disk $D'\subset D$ such that $p$ is
the only point of $\Phi(\{0\}\times\overline{D'})$ mapped to $w_0(p)$.
For sufficiently small $\varepsilon'>0$, set
\[
    \Omega:=\Phi((-\varepsilon',\varepsilon')\times D').
\]
The homotopy
\[
    H_\lambda(t,z)
      :=T(p)+t\ell+w\bigl(\Phi((1-\lambda)t,z)\bigr),
    \qquad 0\le\lambda\le1,
\]
joins $T\circ\Phi$ to the product map
\[
    (t,z)\longmapsto T(p)+t\ell+w\bigl(\Phi(0,z)\bigr).
\]
The homotopy $H_\lambda$ avoids $T(p)$
on the boundary of $\Omega$. Indeed, if
$H_\lambda(t,z)=T(p)$, then the orthogonal decomposition
$\mathbb R^3=\mathbb R\ell\oplus\ell^\perp$ first gives $t=0$ and then
$
    w_0(\Phi(0,z))=w_0(p)
$.

By homotopy invariance of the Brouwer degree
\cite[p.~147]{OutereloRuiz}, we obtain
\[
    \deg(T,\Omega,T(p))=i(p,w_0).
\]

Choose $\delta>0$ such that
\[
    B_\delta(T(p))\cap T(\partial\Omega)=\varnothing,
\]
and let
$
    x\in B_\delta(T(p))\setminus\mathcal P
$
be a globally regular value of $T$. Degree invariance gives
\[
    \deg(T,\Omega,x)=i(p,w_0).
\]
Since $\det DT>0$ at every regular preimage, all local degrees are
positive. Proposition~\ref{prop:degreeone} then yields
\[
    1=\deg_{\mathcal P}(T,x)
      \ge\deg(T,\Omega,x)
      =i(p,w_0)\ge2,
\]
a contradiction.

\textbf{Conclusion.} 
Therefore, $w_0$ is constant, and $\Sigma_0\subset T^{-1}(T(p))$. The regularity of $\Sigma_0$ follows from $\nabla s=\ell\cdot DT$ being $C^2$.

\end{proof}

The following proposition rules out branching of the Tod fibers along
$Z_M$.

    \begin{proposition}
\label{prop:fibers}
For every $x\in T(Z_{M})$, each connected component of
\begin{equation}\label{eq:zero-fibre}
 F_{x}:=Z_{M}\cap T^{-1}(x)
\end{equation}
is a compact, embedded, two-sided surface on which $T$ is constant.
\end{proposition}

\begin{proof}
The set $F_x$ is compact because it is closed in the compact set
$Z_M$. Fix $p\in F_x$ and use the notation from the proof of
Theorem~\ref{local plaque}. In particular,
\[
    s=\ell \cdot\bigl(T-T(p)\bigr)
\]
is a submersion near $p$, and there is a smooth embedded disk
\[
    p\in\Sigma\subset Z_M\cap T^{-1}(T(p))=F_x.
\]
Since $T$ is constant on $\Sigma$,
\[
    T_p\Sigma
    =\ker (DT)_p
    =\ker(ds)_p
    =T_p\{s=0\},
\]
where $T_p$ denotes the tangent space at $p$. 
Therefore, after shrinking the neighborhood $U$ and the plaque $\Sigma$, the
inverse function theorem gives
\[
    \Sigma=\{s=0\}\cap U.
\]
On the other hand, if $p'\in F_x\cap U$, then
$T(p')=x=T(p)$, and hence $s(p')=0$. Thus
\[
    F_x\cap U\subset\{s=0\}\cap U=\Sigma.
\]
The reverse inclusion follows from the definition of $\Sigma$, so
$
    F_x\cap U=\Sigma.
$
Hence, $F_x$ is a smooth embedded surface without boundary.

Each connected component of $F_x$ is closed in the compact set $F_x$
and is therefore compact. Moreover,
\[
    T_pF_x=\ker (DT)_p=(X^\perp)_p.
\]
Since $|X|=N>0$ on $Z_M$, the vector field $X/N$ is a globally defined
unit normal along $F_x$. Consequently, every connected component of
$F_x$ is two-sided.
\end{proof}

\subsection{Exclusion of the null set}

Fix $x_0\in T(Z_M)$ and a connected component $\Sigma$ of the fiber
$F_{x_0}$. By Proposition~\ref{prop:fibers}, $\Sigma$ is a compact embedded
surface without boundary. Along $\Sigma$,
$
    \nu:=N^{-1}X$
is a globally defined unit normal.

Let $\rho$ be the signed tubular coordinate determined by
$\partial_\rho=\nu$ on $\Sigma$. Since $T(\Sigma)=x_0$ and $\mathcal V=0$ on
$\Sigma$, Taylor expansion gives, uniformly on $\Sigma$,
\begin{equation}\label{eq:collar-expansion}
    T=x_0+\rho \mathbf{a}(y)+O(\rho^2),
    \qquad
    \mathcal V=\rho \mathbf{b}(y)+O(\rho^2),
\end{equation}
where $y\in \Sigma$ and
\[
    \mathbf{a}=DT(\nu)|_\Sigma,
    \qquad
    \mathbf{b}=\partial_\nu \mathcal{V}|_\Sigma.
\]
The identity
\[
    T^*\delta=|\mathcal{V}|^2g+X^\flat\otimes X^\flat
\]
and $X=N\nu$ on $\Sigma$ imply
    $|\mathbf{a}|=N>0.$

\begin{lemma}
Let 
\[\widehat{\mathbf{a}}:=\frac{\mathbf{a}}{|\mathbf{a}|}:\Sigma\longrightarrow S^2.\]
Then
\begin{equation}\label{eq:half-collar-degree}
    \operatorname{deg}(\widehat{\mathbf{a}})
    =
    \frac{1}{4\pi}
    \int_\Sigma
    \frac{|\partial_\nu \mathcal{V}|^2}{N^2}\,dA_\Sigma
    \in\mathbb Z_{\ge0}.
\end{equation}
\end{lemma}

\begin{proof}
Orient $\Sigma$ by
\[
    d\rho\wedge dA_\Sigma=d\mu_g
\]
and let $(y^1,y^2)$ be positively oriented local coordinates on
$\Sigma$. From \eqref{eq:collar-expansion},
\[
    \partial_\rho T=\mathbf{a}+O(\rho),
    \qquad
    \partial_iT=\rho\partial_i \mathbf{a}+O(\rho^2),
\]
and therefore
\begin{equation}\label{eq:jacobian-expansion}
    \det(\partial_\rho T,\partial_1T,\partial_2T)
    =
    \rho^2\det(\mathbf{a},\partial_1\mathbf{a},\partial_2\mathbf{a})+O(\rho^3).
\end{equation}
On the other hand,
\[
    |\mathcal{V}|^2=\rho^2|\mathbf{b}|^2+O(\rho^3).
\]
Using the signed Jacobian identity
\[
    T^*(dx^1\wedge dx^2\wedge dx^3)
    =
    N|\mathcal{V}|^2\,d\mu_g,
\]
and comparing the coefficients of $\rho^2$, we obtain
\begin{equation}\label{eq:a-determinant}
    \det(\mathbf{a},\partial_1\mathbf{a},\partial_2\mathbf{a})\,dy^1dy^2
    =
    N|\mathbf{b}|^2\,dA_\Sigma.
\end{equation}
Since
\[
    \widehat{\mathbf{a}}^{\,*}dA_{S^2}
    =
    \frac{\det(\mathbf{a},\partial_1\mathbf{a},\partial_2\mathbf{a})}
         {|\mathbf{a}|^3}\,dy^1dy^2,
\]
$|\mathbf{a}|=N$ and equation \eqref{eq:a-determinant} give
\[
    \widehat{\mathbf{a}}^{\,*}dA_{S^2}
    =
    \frac{|\mathbf{b}|^2}{N^2}\,dA_\Sigma.
\]
Integration proves \eqref{eq:half-collar-degree}.
\end{proof}
Next, we study the half collar regions around $\Sigma$. Set
\[
\Omega_\varepsilon^+:=\{0<\rho<\varepsilon\},
    \qquad
    \Omega_\varepsilon^-:=\{-\varepsilon<\rho<0\}.
\]

\begin{lemma} \label{lem:half-collar}
For every sufficiently small $\varepsilon>0$, there exists $\delta>0$
such that
\begin{equation}\label{eq:two-half-degrees}
    \deg(T,\Omega_\varepsilon^\pm,x)
    =
    \deg(\widehat{\mathbf a})
\end{equation}
for every
$
    x\in B_\delta(x_0)\setminus\{x_0\}$. 
    Moreover, if $\deg(\widehat{\mathbf a})=0$, then there exist
$\varepsilon_+,\varepsilon_->0$ such that
\[
    \mathcal V\equiv0
    \quad\text{on}\quad
    \Omega_{\varepsilon_+}^+
    \ \text{and}\ 
    \Omega_{\varepsilon_-}^-.
\]
\end{lemma}
\begin{proof}
Since $\Sigma$ is compact and $|\mathbf a|=N>0$ on $\Sigma$, the first equation in
\eqref{eq:collar-expansion} implies that, for every sufficiently small $\varepsilon>0$, $T(\{\rho=\pm \varepsilon\})$ stays a positive distance from $x_0$. Choose
$\delta>0$ smaller than both distances.

Fix a regular value
$   x\in B_\delta(x_0)\setminus\{x_0\}$. 
On the inner boundary $\{\rho=0\}$, the normalized map
\[
    \frac{T-x}{|T-x|}:\; \partial\Omega^+_\varepsilon\to S^2
\]
is constant and therefore has degree zero. On the outer boundary
$\{\rho=\varepsilon\}$ of $\Omega_\varepsilon^+$, first
by moving the target from $x$ to $x_0$ and then by using
\eqref{eq:collar-expansion}, it is homotopic to $\widehat{\mathbf a}$. By the boundary representation of the Brouwer degree
\cite[p.~157]{OutereloRuiz},
\[
    \deg(T,\Omega_\varepsilon^+,x)
    =\deg\left( \frac{T-x}{|T-x|}\right)
    =\deg(\widehat{\mathbf a}).
\]

For $\Omega_\varepsilon^-$, the inner boundary $\{\rho=-\varepsilon\}$ has the
opposite boundary orientation, while its normalized leading term is
$-\widehat{\mathbf a}$. Since the antipodal map on $S^2$ has degree
$-1$, these two signs cancel. Hence
\[
    \deg(T,\Omega_\varepsilon^-,x)
    =
    \deg(\widehat{\mathbf a}),
\]
which proves \eqref{eq:two-half-degrees}.

When $x$ is a critical value, we can approximate $x$ by regular values to obtain \eqref{eq:two-half-degrees}.

Suppose now that $\deg(\widehat{\mathbf a})=0$ and fix one side of
$\Sigma$. Assume, toward a contradiction, that $\mathcal V$ does not
vanish on any smaller half-collar on that side. Then there are points
arbitrarily close to $\Sigma$ at which $\mathcal V\ne0$. At each such
point,
\[
    \det DT=N|\mathcal V|^2>0,
\]
so $DT$ is invertible. By the inverse function theorem and Sard's
theorem, one can choose a globally regular value
\[
    x\in B_\delta(x_0)\setminus
       \bigl(\{x_0\}\cup\mathcal P\bigr)
\]
having a preimage in that half-collar. Every preimage there has
positive local degree, and at least one such preimage exists.
Consequently,
\[
    \deg(T,\Omega_\varepsilon^\pm,x)>0,
\]
contradicting \eqref{eq:two-half-degrees}. Thus $\mathcal V$ vanishes
on a possibly smaller half-collar on the chosen side. Applying the
same argument to the other side completes the proof.
\end{proof}

\begin{theorem}\label{thm:no-null-set}
Under the hypotheses of Theorem \ref{T:main} and $\mathcal{E}>|P|$, we have
\[
    Z_M=\varnothing.
\]
Consequently, $v$ is nowhere zero.
\end{theorem}

\begin{proof}
Suppose $Z_M\ne\varnothing$. Fix $x_0\in T(Z_M)$ and let $\Sigma$ be
any connected component of the null fiber $F_{x_0}$. If
$\operatorname{deg}(\widehat{\mathbf{a}})\ge1$, choose $\varepsilon$, $\delta$, and a globally regular
value
$
    x\in B_\delta(x_0)\setminus
       \bigl(\{x_0\}\cup\mathcal P\bigr)
$
as in Lemma~\ref{lem:half-collar}. Each of the two disjoint
half-collars has relative degree $\operatorname{deg}(\widehat{\mathbf{a}})$. Since every regular
preimage has positive sign,
\[
    1=\deg_{\mathcal P}(T,x)
      \ge 2\operatorname{deg}(\widehat{\mathbf{a}})
      \ge2,
\]
which is impossible. Hence, $\operatorname{deg}(\widehat{\mathbf{a}})=0$ for every
component of $F_{x_0}$.

Applying Lemma~\ref{lem:half-collar} on both sides of each such
$\Sigma$, it follows that every point of $Z_M$ has a neighborhood
contained in $Z_M$. Thus $Z_M$ is open. It is also closed, since
$Z_M$ is compact.  Therefore, we conclude that
$Z_M=\varnothing$.
\end{proof}

\appendix

\section{Embedding into a Majumdar-Papapetrou spacetime}

We summarize the known results showing how the non-existence of super-covariantly constant spinors of mixed causal types implies Theorem \ref{T:main}.

\subsection{The timelike case}

Suppose that $\psi$ is timelike everywhere.
First, we show that $(M,g,k)$ embeds into a so-called Israel-Wilson-Perj\'es (IWP) spacetime.
This goes back to Tod \cite{Tod}.
Consider the Killing development 
\begin{align*}
    \overline g=&
    -v^2dt^2+g+2X^\flat\otimes dt\\
    =&-v^2(dt-v^{-2}X^\flat)^2+g+v^{-2}X^\flat\otimes X^\flat.
\end{align*}
Using the property
    \begin{align*}
        T^\ast \delta =v^2g+X^\flat\otimes X^\flat,
    \end{align*}
we find that $\overline g$ describes an IWP spacetime.

\medskip

Next, we show that the above IWP spacetime must be of the Majumdar-Papapetrou family.
This argument is due to Chru\'sciel-Real-Tod \cite{CRT}.
Let $H=\frac{L+iK}{v^2}$. Note that $H$ is harmonic with respect to $\delta $.
Write $H=Ue^{i\theta}$.
Then $\omega=v^{-2}X^\flat$ satisfies
\begin{align*}
    \operatorname{curl}_\delta(\omega)=2U^2\nabla \theta.
\end{align*}
Note that this implies
\begin{align*}
    \operatorname{div}_\delta(U^2\nabla\theta)=0.
\end{align*}
Next, we integrate this identity over $M$.
Integration by parts yields
\begin{align*}
    \int_{\mathbb{R}^3} U^2|\nabla\theta|^2=0.
\end{align*}
Note that due to the asymptotics, there are no boundary contributions coming from either the asymptotically flat end or the cylindrical ends.
Consequently, $\theta$ is constant and $\overline g$ describes an MP metric.

\subsection{The null case}

Suppose that $\psi$ is null everywhere.
In this case, the argument follows exactly along the lines of Hirsch-Zhang \cite{HirschZhang}.
Although $E,B$ need not vanish, we still have $dX^\flat=0$ and obtain a foliation, and $\psi|\psi|^{-1}$ is parallel on each leaf of this foliation.
Note that $\psi$ does not correspond to a parallel spinor in the pp-wave spacetime, although its restriction is parallel along each null hypersurface; the electromagnetic fields rotate it as one moves in the wave direction.

\end{document}